\documentclass[11pt]{article}
\usepackage[margin=1in]{geometry}

\usepackage{amsmath,amssymb,amsthm,amsfonts,cite,alltt,clrscode,graphicx,color}
\usepackage{bold-extra}
\usepackage{microtype}
\usepackage{hyperref}
\usepackage{url}
\usepackage[american]{babel} 
\let\epsilon=\varepsilon

\newtheorem{theorem}{Theorem}
\newtheorem{lemma}[theorem]{Lemma}

{\makeatletter
 \gdef\xxxmark{%
   \expandafter\ifx\csname @mpargs\endcsname\relax 
     \expandafter\ifx\csname @captype\endcsname\relax 
       \marginpar{xxx}
     \else
       xxx 
     \fi
   \else
     xxx 
   \fi}
 \gdef\xxx{\@ifnextchar[\xxx@lab\xxx@nolab}
 \long\gdef\xxx@lab[#1]#2{\textbf{[\xxxmark #2 ---{\sc #1}]}}
 \long\gdef\xxx@nolab#1{\textbf{[\xxxmark #1]}}
\long\gdef\xxx@lab[#1]#2{}\long\gdef\xxx@nolab#1{}%
}

\newcommand{\In}{\textsc{In}}
\newcommand{\Out}{\textsc{Out}}
\newcommand{\inin}{\textsc{In-In}}
\newcommand{\inout}{\textsc{In-Out}}
\newcommand{\outout}{\textsc{Out-Out}}

\title{Continuous Flattening of All Polyhedral Manifolds \\
       using Countably Infinite Creases}

\author{%
  Zachary Abel%
    \thanks{MIT EECS Department}
\and
  Erik D. Demaine%
    \thanks{MIT Computer Science and Artificial Intelligence Laboratory}
\and
  Martin L. Demaine\footnotemark[2]
\and
  Jason~S.~Ku%
    \footnotemark[1]
\and
  Jayson Lynch%
    \footnotemark[2]
\and
  Jin-ichi Itoh%
    \thanks{School of Education, Sugiyama Jogakuen University}
\and
  Chie Nara%
    \thanks{Meiji Institute for Advanced Study of Mathematical Sciences,
      Meiji University}
}

\date{}

\begin{document}

\maketitle

\begin{abstract}

We prove that any finite polyhedral manifold in 3D can be continuously
flattened into 2D while preserving intrinsic distances and avoiding crossings,
answering a 19-year-old open problem, if we extend standard folding models to
allow for countably infinite creases.
The most general cases previously known to be continuously
flattenable were convex polyhedra and semi-orthogonal polyhedra.
For non-orientable manifolds, even the existence of an instantaneous
flattening (flat folded state) is a new result.
Our solution extends a method for flattening semi-orthogonal polyhedra:
slice the polyhedron along parallel planes and flatten the polyhedral strips
between consecutive planes.
We adapt this approach to arbitrary nonconvex polyhedra by generalizing strip
flattening to nonorthogonal corners and slicing along a countably
infinite number of parallel planes, with slices densely approaching every
vertex of the manifold.  We also show that the area of the polyhedron that
needs to support moving creases (which are necessary for closed polyhedra
by the Bellows Theorem) can be made arbitrarily small.

\end{abstract}


\section{Introduction}

We crush polyhedra flat all the time, such as when we recycle cereal boxes or
store airbags in a steering wheel.  But is this actually possible without
tearing or stretching the material?  This problem was first posed in 2001
\cite{Demaine-Demaine-Lubiw-2001-flattening-manuscript} (see
\cite[Chapter~18]{GFALOP}): does every polyhedron have a continuous motion that
preserves the metric (intrinsic shortest paths), avoids crossings, and ends in a
flat folded state? This problem is Open Problem 18.1 of the book \emph{Geometric
Folding Algorithms} \cite{GFALOP}. In this paper, we solve this 19-year-old open
problem with a positive answer: every polyhedron can be continuously flattend.
Specifically, we prove for a broad definition of \emph{polyhedron}: any compact
polyhedral 2-manifold (possibly with boundary) embedded in 3D and having
finitely many polygonal faces. However, our result is arguably in a model not
intended by the original problem: our folding has countably infinitely many
creases at all times.

A necessary first step is to show that every polyhedron has a flat folded
state (the end of the desired flattening motion).  This problem was also
first posed in 2001 \cite{Demaine-Demaine-Lubiw-2001-flattening-manuscript},
where it was solved for convex and semi-orthogonal polyhedra.%
\footnote{In a \emph{semi-orthogonal} polyhedron, every facet is
  either parallel or perpendicular to a common plane.  Thus,
  in some orientation, the faces are all horizontal (parallel to the floor)
  or vertical (perpendicular to the floor).}
Later, Bern and Hayes \cite{Bern-Hayes-2011-flattening}
solved the problem for \emph{orientable} polyhedral manifolds,
generalizing a previous solution for sphere or disk topology
\cite{GFALOP}.
This result solved Open Problem 18.2 of~\cite{GFALOP}
(also originally posed in 2001
\cite{Demaine-Demaine-Lubiw-2001-flattening-manuscript}),
except for non-orientable polyhedral manifolds, which we solve here.

Continuous flattening necessarily requires continuously moving/sliding the
creases on the surface over time (for polyhedra enclosing a volume):
if all creases remained fixed throughout the motion
(and the set of creases is finite), then
the Bellows Theorem \cite{Connelly-Sabitov-Walz-1997}
tells us that the volume would remain fixed, so could not decrease to zero.
A natural question, though, is how much area of the surface needs to be
flexible in the sense of supporting moving creases, and how much can be
made of rigid panels connected by hinges.
Abel et al.~\cite{abel2015rigid} showed that a surprisingly small but finite
slit suffices for continuous flattening of a regular tetrahedron.
Matsubara and Nara \cite{nara2016infimum} recently showed that an arbitrarily
small area of flexibility suffices for $\alpha$-trapezoidal polyhedra.
In this paper, we show that an arbitrarily small area of flexibility
suffices for any polyhedral manifold.

Several previous results constructed continuous flattenings of special classes
of polyhedra.  Itoh and Nara \cite{Itoh-Nara-2011} solved Platonic solids while
preserving two faces, and later with V\^{\i}lcu \cite{itoh2011continuous} solved
convex polyhedra using Alexandrov surgery (which is difficult to compute).  At
SoCG 2014, Abel et al.~\cite{abel2014continuously} solved convex polyhedra using
a simple algorithm that respects the \emph{straight skeleton gluing},
corresponding to the intuitive way to flatten a polyhedron, and solving Open
Problem 18.3 of \cite{GFALOP} (the last open problem of Chapter~18, also
originally posed in 2001
\cite{Demaine-Demaine-Lubiw-2001-flattening-manuscript}). Unfortunately, this
approach seems difficult to extend to nonconvex polyhedra. More recently, a
slicing approach (dating back to
\cite{Demaine-Demaine-Lubiw-2001-flattening-manuscript}) was shown to
continuously flatten semi-orthogonal polyhedra
\cite{FlatteningOrthogonal_JCDCGG2015full}. In this paper, we extend this
slicing approach in several ways to solve arbitrary polyhedral manifolds.


\subsection{Approach}

We generalize the slicing approach of \cite{FlatteningOrthogonal_JCDCGG2015full},
which conceptually cuts the polyhedron along parallel planes through every
vertex, and several additional planes in between so that the resulting
\emph{slabs} (portions of the polyhedron between consecutive planes) are
``short''.
In \cite{FlatteningOrthogonal_JCDCGG2015full}, each slab is an orthogonal
band, which is relatively easy to flatten continuously. The key difference in
our case is that the slabs are much more general: in general, a slab in a
polyhedron is a \emph{prismatoid} (excluding the top and bottom faces),
that is, a polyhedron whose vertices lie in two parallel planes,
whose faces are triangles and trapezoids spanning both planes.
Unfortunately, prismatoids seem extremely difficult to flatten continuously,
as original polyhedron vertices are particularly difficult to handle
in the general case.

To circumvent this challenge, we instead target the flattening of
\emph{prismoids}: prismatoids whose spanning faces are only trapezoids having
parallel top and bottom edges (i.e., no triangular spanning faces), where every
vertex is incident to at most two spanning trapezoids. We will use the term
\emph{cylindrical prismoid} to refer to the spanning faces of a prismoid,
without the top and bottom face. A key innovation in our approach is to divide a
polyhedral manifold using \emph{countably infinitely} many parallel planar cuts,
with slabs approaching zero height as we approach polyhedron vertices. As a
result, \emph{all} slabs consist of disjoint cylindrical prismoids.
The key property is that original polyhedron vertices do not appear on the
boundary of \emph{any} slab, because any such slab would get divided in half
through countably infinite recursion. 

Note that since we allow polyhedral
manifolds with boundary, a component within a slab may only be a subset of a
cylindrical prismatoid, we call a \emph{prismoidal wall}; this
generalization is discussed in Section~\ref{sec:slice}.


\subsection{Outline}

We implement the approach described above in a bottom-up fashion.
First, Section~\ref{sec:model} formally defines our model of folding. Next, Section~\ref{sec:collapse} shows how to
collapse prismoid edges and faces by constructing generalized \inout\ and
\outout\ gadgets.
Then, Section~\ref{sec:slice} shows how to slice the input
polyhedral manifold so that we can flatten subsets of it using the methods from
Section~\ref{sec:collapse}.
Finally, Section~\ref{sec:proof} puts these algorithms together to prove
the following theorem:

\begin{theorem}
\label{thm:flatten}
Any compact polyhedral 2-manifold (possibly with boundary) embedded in 3D
and having finitely many polygonal faces can be continuously
flattened while preserving intrinsic distances and avoiding crossings.
A flattening motion exists such that at all times during the flattening motion (except the beginning),
the folded form consists of countably infinitely many creases,
with finitely many accumulation lines.
Furthermore, the area supporting moving creases can be made arbitrarily small.
\end{theorem}   


\section{Model}
\label{sec:model}

The standard model of folding 2D surfaces in 3D \cite[Chapter~11]{GFALOP}
assumes finitely many creases, as that is the primary case of interest for
origami.  A full definition supporting countably infinitely many creases
is likely possible, but difficult, as it is no longer possible to focus on
well-behaved positive-area neighborhoods.
For the purposes of this paper, we define a limited model of folding with
countably infinite creases, where the folding decomposes into components
separated by horizontal planes, and each component is a finite-crease
folding according to \cite[Chapter~11]{GFALOP}.

Specifically, define a \emph{stacked folded state} of a polygon $P$ of paper
to consist of two components:
\begin{enumerate}

\item A decomposition of $P$ into countably many topologically closed polygonal
  regions $P_1, P_2, \dots$ (the unfolded slices, each of which can be
  disconnected), whose interior-disjoint union $\cup_{i=1}^\infty P_i$
  equals~$P$.
  The sequence $P_1, P_2, \dots$ can be (countably) infinite,
  and is in no particular order
  (in particular, it does not match the stacking order defined below
  in Property~\ref{ordered}).
  Because of the infinite decomposition, some points of $P$ belong to
  one or two $P_i$ (two in the case of shared boundary),
  while other points of $P$ may not belong to any $P_i$
  but rather exist in the limit of some sequence $P_{k_1}, P_{k_2}, \dots$.
\item A finite-crease folded state $(f_i,\lambda_i)$ of each region $P_i$
  (the folded slice),
  consisting of a geometry $f_i : P_i \to \mathbb R^3$ and a layer-ordering
  partial function $\lambda_i : P_i^2 \to \{-1,+1\}$
  (as in \cite[Chapter~11]{GFALOP}).
\end{enumerate}
These components must satisfy the following constraints:
\begin{enumerate}
\setcounter{enumi}2
\item \label{ordered}
  The decomposition $P_1, P_2, \dots$ has a total ordering $\prec$ for which
  each $P_i$ intersects only its immediate predecessor and successor in~$\prec$.
  \xxx{Is this right? A little confused about the limit points...}
\item \label{meet}
  The folded states meet on their shared boundaries, i.e.,
  $f_i(P_i \cap P_{j}) = f_{j}(P_i \cap P_{j})$ for~all~$i, j$. 
\item \label{limits meet}
  For every point $q \in P$, there is a unique point $r \in \mathbb R^3$
  (more naturally notated $f(q)$) such that,
  for every sequence $q_1, q_2, \dots$ of points in $P$
  converging to~$q$,
  if each $q_i$ belongs to a corresponding region $P_{k_i}$,
  then sequence $f_{k_i}(q_i)$ converges to $r$.
  This property guarantees a global folded-state geometry $f$ on all of~$P$,
  in particular for points that do not belong to any~$P_i$.
\item \label{slices}
  The folded states live in interior-disjoint horizontal slices of space,
  i.e., all points in $f_i(P_i)$ have $z$ coordinates in the range
  $Z_i = [z_i^-, z_i^+]$, where the intervals $Z_0, Z_1, Z_2, \dots$
  are interior-disjoint and $P_i \prec P_j$ implies $z_i^+ \leq z_j^-$.
\end{enumerate}
Intuitively, these constraints guarantee that there are no proper collisions
between different folded states $(f_i,\lambda_i)$ where they join.
Although two folded states may touch in a shared horizontal plane, the
total ordering from Property~\ref{ordered} provides a stacking order for
such overlapping layers.
A subtlety here is that, to allow the final flat folded state
where $Z_1 = Z_2 = \cdots = [z,z]$ for some $z$, we need to allow each
interval $Z_i$ to degenerate to a point, allowing for potentially many folded
states to overlap in that single $z$ plane
(without violating interior-disjointness of Property~\ref{slices}).


With this notion in hand, we can define \emph{(stacked) folding motions} as in
\cite[Chapter~11]{GFALOP}: a continuous function $M$ mapping each time
$t \in [0,1]$ to a stacked folded state, where each $P_i(t)$ and $z_i(t)$
varies continuously with time, and the restriction of $M$ to each $P_i(t)$
produces a valid folding motion of the finite-crease folded state
$(f_i(t),\lambda_i(t))$.
A \emph{(stacked) flattening motion} is a (stacked) folding motion $M$
such that the final folded state $M(1)$ lies in a single $z$ plane.


\section{Flattening Prismoids} 
\label{sec:collapse}

In this section, we show how to flatten prismoids which have a small height
relative to their other features. We will then use this technique to flatten
arbitrary polyhedral manifolds after slicing them into a countable set of such
prismoids, as detailed in Section~\ref{sec:slice}.
Section~\ref{sec:PrismoidOverview} describes an overview of our approach, and
the rest of this section describes the details of how to locally flatten the
edges and faces of a prismoid.

\subsection{Approach}
\label{sec:PrismoidOverview}

To specify the approach in more detail, let us recall the overall approach
for semi-orthogonal polyhedra from \cite{FlatteningOrthogonal_JCDCGG2015full}.
Call a prismoid edge \emph{spanning} if its endpoints lie in the top and
bottom planes, and a prismoid face \emph{spanning} if it includes vertices
in both the top and bottom planes.
Unlike in the Introduction, here we include the top and bottom horizontal faces as
part of the prismoid (which will later represent attachments to neighboring
prismoids), and we will continuously flatten while moving the horizontal faces only vertically.
The approach of~\cite{FlatteningOrthogonal_JCDCGG2015full} flattens orthogonal
spanning edges of a (possibly nonconvex) prism using two methods.
One method bends
both faces adjacent to a spanning edge \emph{toward} the convex side of the edge, while
the other method bends one face \emph{toward}, and one face \emph{away}, from the
convex side; we
call these general strategies \inin\ and \inout\ respectively. In each method,
the top face is translated down normal to the face onto the bottom face; faces
adjacent to the edge are bent using a single crease far from each spanning edge;
while additional local creases are added in order to collapse each edge. This
strategy allows a common interface between spanning edge collapsing crease
patterns so each edge can be dealt with independently, assuming the edges are
far enough apart. By alternately labeling each
face around the prism as \In\ or {\Out},
each spanning edge can then be collapsed using their \inout\
method, with the exception of perhaps one spanning edge collapsed using their
\inin\ method; see Figure~\ref{fig:prism}.

\begin{figure}
	\includegraphics[width=\linewidth]{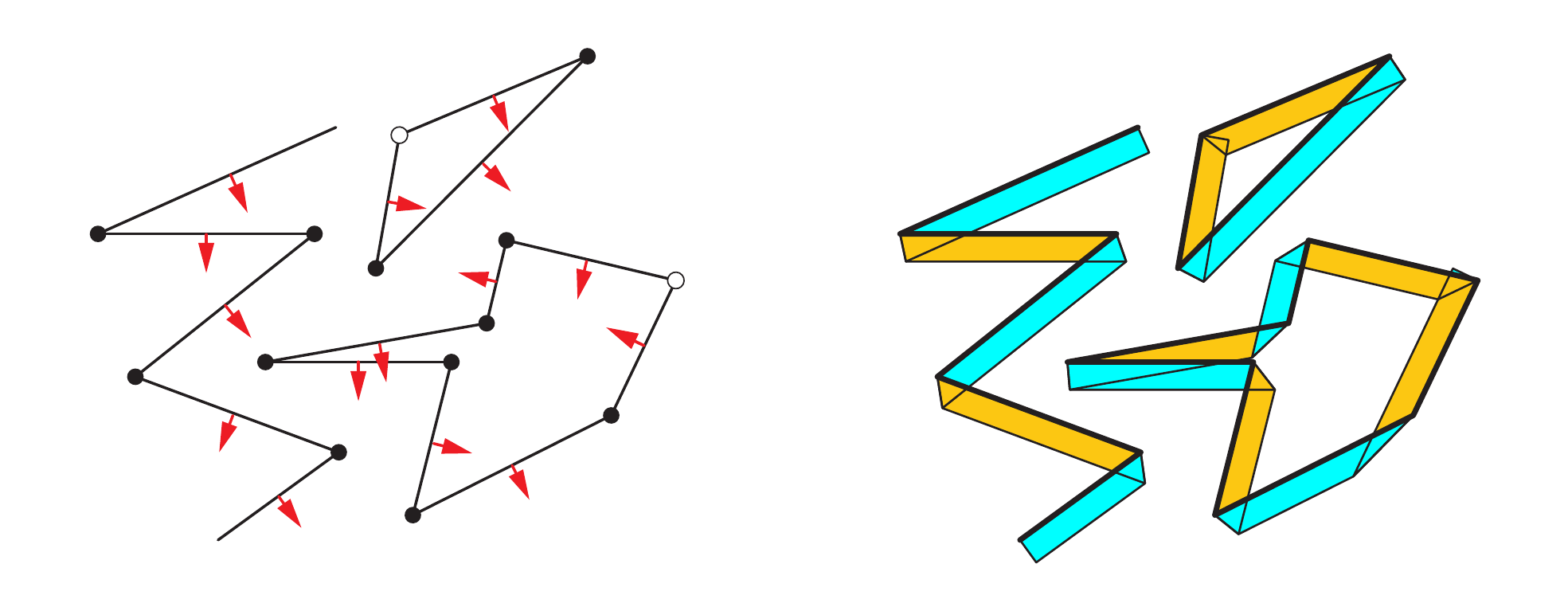}
	\caption{[Left] Top view of a semi-orthogonal set of walls, assigning a
		direction to each edge and labeling non-terminating vertices as either \inin\ in
		white or \inout\ in black. 
		[Right] The flattened state associated with this direction assignment.}
	\label{fig:prism}
\end{figure}

Our edge flattening construction generalizes their orthogonal approach for
nonorthogonal edges, allowing us to flatten general prismoids. There are a few
key differences between the gadgets presented here and the gadgets presented
in~\cite{FlatteningOrthogonal_JCDCGG2015full}. While we give a construction for a
generalized \inout\ gadget, we provide a construction for an \outout\ gadget,
bending both faces adjacent to an edge away from the convex side, instead of an
\inin\ gadget. While their orthogonal \inout\ gadget constructs three new crease
pattern vertices at any intermediate folded state to flatten each spanning edge,
our generalized \inout\ gadget requires construction of only two new vertices,
simplifying the structure. Additionally, our \outout\ gadget has the same
topological complexity as the orthogonal \inin\ gadget, both requiring
construction of two new vertices. Lastly, both orthogonal gadgets require some
adjacent faces to be coplanar and touching throughout the folding motion, which
may not be desirable; by contrast, faces in our generalized gadgets never touch
face to face except in the final flat-folded state.

\subsection{Gadget Parameterization}

The next three sections describe how to locally flatten spanning edges of a
prismoid by detailing two gadgets: an \inout\ gadget and an \outout\ gadget.
Because the top and bottom faces of a prismoid must all collapse consistently
and simultaneously, we give a single parameterization for the entire collapse;
see Figure~\ref{fig:setup} [Left]. Of the two prismoid vertices incident to
the spanning edge, at least one has an angular deficit no greater than $\pi$. We
choose such a vertex to be the \emph{primary vertex}, and let $\theta$,
$\alpha$, and $\beta$ be the three face angles incident to it, with $\theta$ the
angle at the base, and $\alpha$ and $\beta$ the two angles of the incident
spanning faces. When we speak locally of a spanning edge, the \emph{primary
vertex} will be vertex $o$, with the other \emph{non-primary vertex} being $q$.
By fixing the spanning edge to have unit length, we can uniquely specify any
prismoid spanning edge up to affine transformations by choosing $\theta$,
$\alpha$, and $\beta$ such that: 

\begin{itemize}
\setlength\itemsep{0pc}
\item $0 < \theta$ because we forbid touching faces in the input polyhedron;
\item $|\alpha - \beta| < \theta < \alpha + \beta$ or else the prismoid is already flat; and 
\item $\alpha+\beta \leq \pi$ as defined for a primary vertex.
\end{itemize}

\begin{figure}
\includegraphics[width=\linewidth]{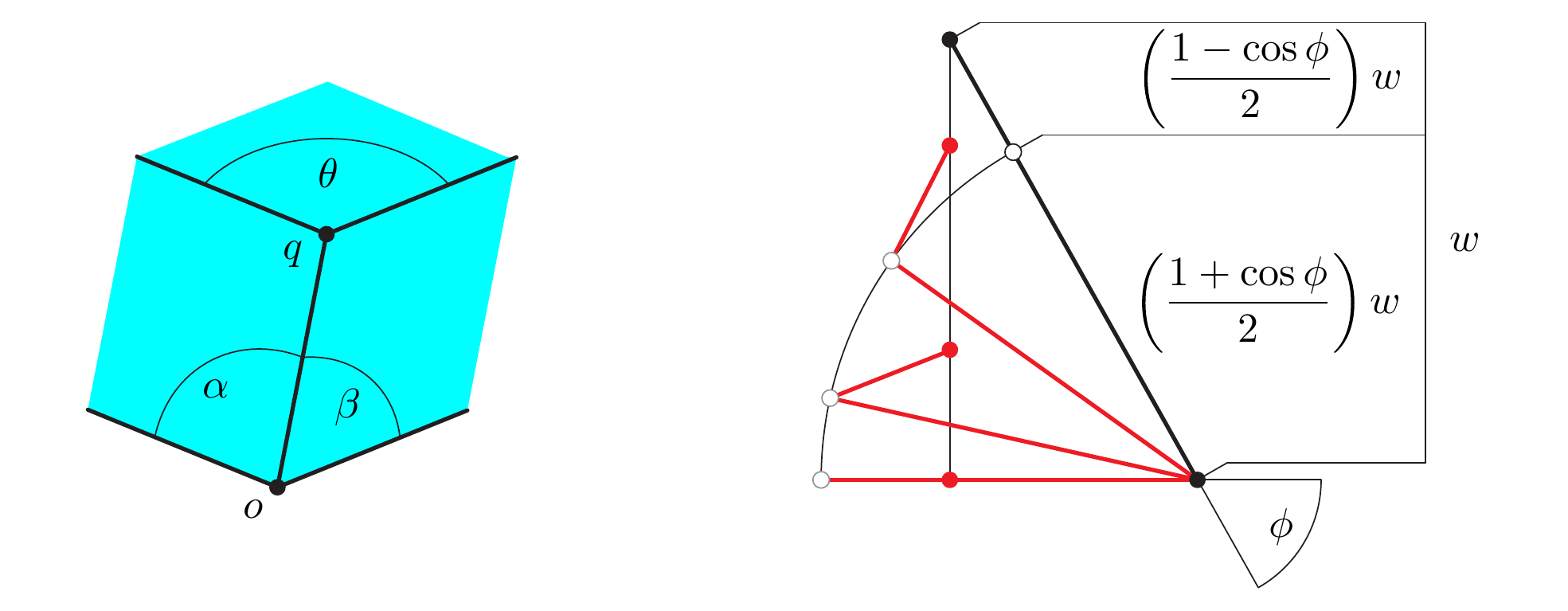}
\caption{[Left] Parameterization of a spanning edge up to affine transformation.
[Right] Cross section of a spanning face flattening along a single crease.}
\label{fig:setup}
\end{figure}

\subsection{Spanning Face Collapse}
\label{sec:facecollapse}

Each spanning face is angled relative to the top and bottom face of the
prismoid. The dihedral angle $\phi$ of the face relative to the base uniquely
determines the crease line that will collapse the face flat when sufficiently
far from a spanning edge. We will translate the top face down normal to the face
onto the bottom face; see Figure~\ref{fig:setup} [Right] for a cross
section of the face collapse. Two different single-crease solutions can allow
this flattening to occur, either flattening the face to one side or the other.
Call the width the distance between the top and bottom edge of the face.
In either case, the crease we introduce will separate the width of the face $w$
into sections of width 
\begin{equation} 
\left(\frac{1\pm \cos\phi}{2}\right)w.
\end{equation}

We call such a crease a \emph{spanning face crease}. The crease will
be closer to the bottom if the face folds toward the bottom edge, and closer to
the top if the face folds toward the top edge. Local to a vertex, we can write
$w$ and $\phi$ on both the $\alpha$ and $\beta$ sides of the edge in terms of
our parameterization:
\begin{align}
w_\alpha &= \sin\alpha,\quad \cos\phi_\alpha =
\csc\alpha(\cos\beta\csc\theta - \cos\alpha\cot\theta), \\
w_\beta &= \sin\beta,\quad \cos\phi_\beta =
\csc\beta(\cos\alpha\csc\theta - \cos\beta\cot\theta).
\end{align}

We note also that collapsing the face along this crease keeps folded material
within distance $(1-\cos\phi)w/2$ of the projection of the face onto the
prismoid base, when $\phi$ and $w$ are strictly positive. Further, it is easy to
verify that during a collapse, spanning face creases always exist between the
top and bottom faces (strictly between except in the final flat-folded state).

\begin{figure}[h!]
\includegraphics[width=\linewidth]{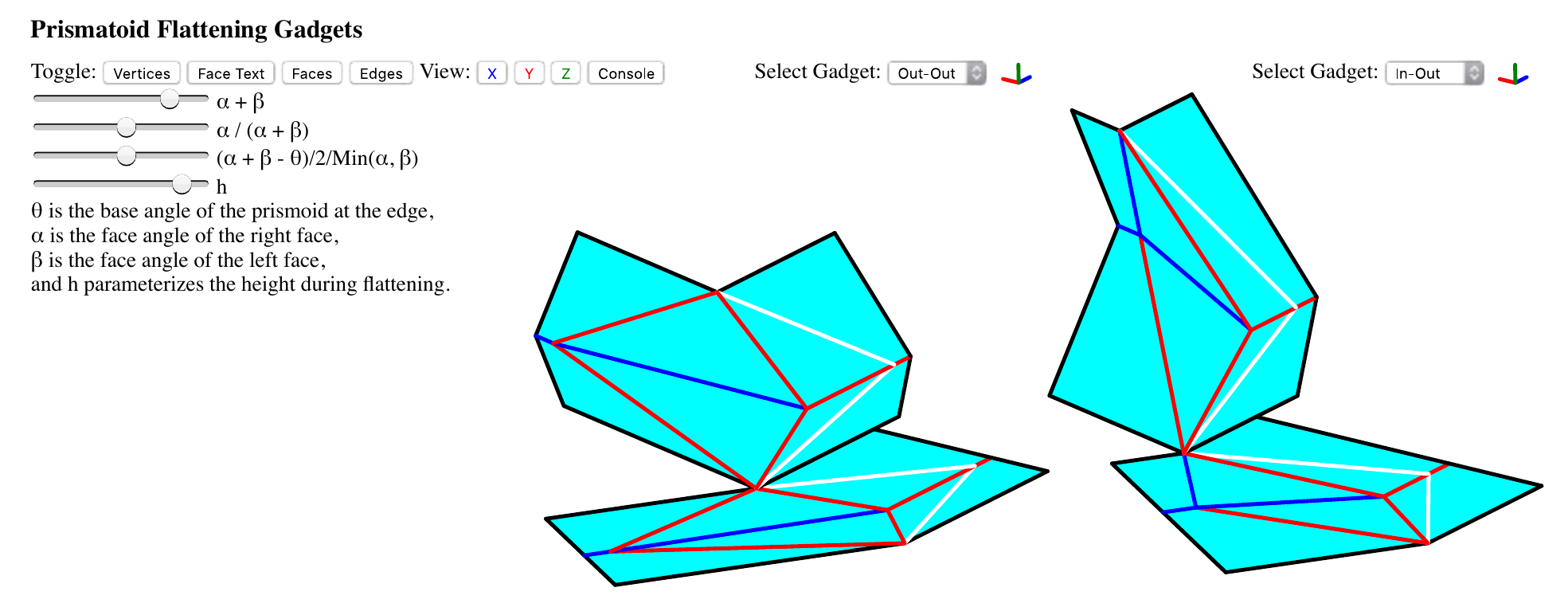}
\caption{View of the web application \cite{webapp} for interacting with the prismoid 
spanning edge flattening gadgets. The \outout\ gadget is shown on the [Left] and
the \inout\ gadget is shown on the [Right].}
\label{fig:webapp}
\end{figure}

\subsection{Interactive Gadget Visualization}

In the following two sections, we describe and analyze our \outout\ and \inout\
gadgets for flattening spanning edges. To supplement understanding, we have
implemented a web application to visualize these gadgets over the parameterized
space of possible spanning edges. You can find it here~\cite{webapp}. The
application is written in CoffeeScript and is open source. Using the app, you
can explore the different spanning edges over the parameterized space as well as
intermediate folded states of the continuous folding motion.
Figure~\ref{fig:webapp} shows a view of the interface and display. Toggling
``Vertices'' will show numeric labels for the vertices. Vertices $\{10, 13, 16,
17\}$ in the animations for both the \outout\ and \inout\ gadgets correspond to
respective points $\{o, q, p_\beta, p_\alpha\}$ in Figures~\ref{fig:outout} and
\ref{fig:inout}.

\begin{figure}
\includegraphics[width=\linewidth]{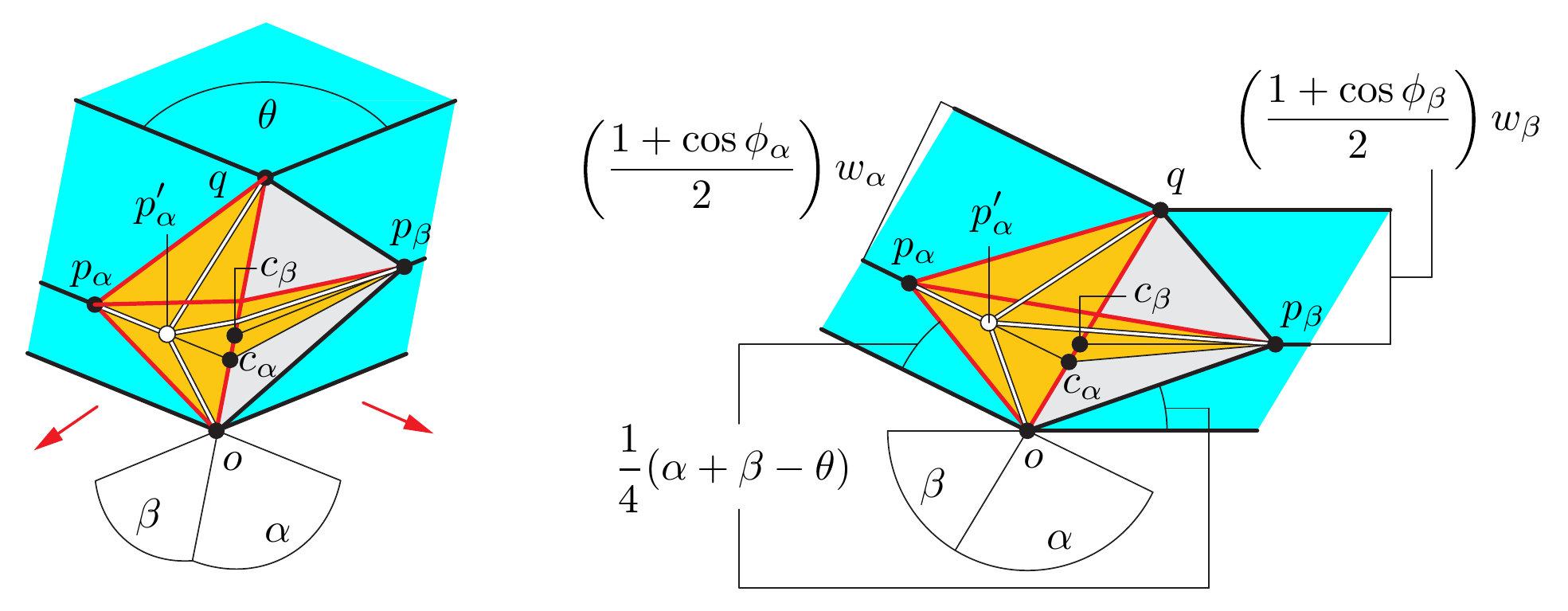}
\caption{Reference points and creases for the \outout\ gadget drawn on [Left]
the surface of a prismoid local to a spanning edge, and [Right] the development
of the spanning faces adjacent to the edge.} 
\label{fig:outout} 
\end{figure}

\subsection{\outout\ Gadget}

We describe the construction of our \outout\ gadget, and then show that it
folds continuously while preserving intrinsic distances and avoiding crossings.
We begin by constructing a flat folded state, introducing two crease pattern
vertices and show that moving one of these vertices along a line provides the
desired folding motion. Figure~\ref{fig:outout} corresponds to our construction
described below. Consider a prismoid spanning edge parameterized by $\theta$,
$\alpha$, and $\beta$. First we construct the spanning face creases on each side
according to the characterization in Section~\ref{sec:facecollapse}, and let
$c_\alpha$ and $c_\beta$ be the locations where respective spanning face creases
meet the spanning edge. Then we can flatten the primary vertex $o$ locally using two
creases so the adjacent faces collapse away from the convex side of the edge. In
order to flatten angle $\alpha + \beta$ of material with two creases while
keeping the bounding edges angle $\theta$ apart, the angle between the creases
must be $(\alpha+\beta+\theta)/2$. Any such creases will suffice for our
construction, but choosing one that is somewhat centered will keep the gadget
closer to the spanning edge. We choose the pair centered on the primary vertex
$o$,
so that the angle between a crease and its adjacent bottom edge is the same,
$(\alpha+\beta-\theta)/4$. Terminate each of these creases when they meet their
respective spanning face crease. Let these termination points be $p_\alpha$ and
$p_\beta$ respectively. Complete the crease pattern by adding creases along the
three pairwise shortest paths between these two points and non-primary vertex 
$q$. Some tedious but straightforward algebra confirms that this crease pattern is
always flat-foldable, with each vertex satisfying Kawasaki's
Theorem~\cite{GFALOP}.

This flat-foldable crease pattern corresponds to a folding mechanism that has a
single-degree of freedom because the internal vertices are nondegenerate and
degree-four. However, when this crease pattern unfolds rigidly, the \emph{base
angle} spanned by the two boundary edges incident to the primary vertex $o$ will
open monotonically from $\theta$ to $\alpha+\beta$ when fully unfolded. As a
thought exercise, let us fix the crease pattern folded to some three-dimensional
intermediate folded state so that the base angle is strictly between $\theta$
and $\alpha+\beta$, and then remove the two triangular faces from the crease
pattern. We have removed a quadrilateral of material that was creased from
$p_\alpha$ to $p_\beta$, i.e.\ in the folded state, the distance between
$p_\alpha$ and $p_\beta$ is the same as when the material is unfolded. Now we
rotate the bent $\alpha$ and $\beta$ spanning faces together around the axis
from $o$ to $q$ until the base angle is $\theta$, which brings points $p_\alpha$
and $p_\beta$ closer together. What remains is a folding of a subset of the
prismoid corner that matches the top and bottom face angles in an intermediate
folded state. It remains to replace the hole with the quadrilateral of paper we
removed. Noting that nonadjacent vertices of quadrilateral hole are now strictly
contractive in this intermediate folded state, we appeal to the construction
in~\cite{holefilling} to construct an isometry. Extend the spanning face crease
incident to $p_\alpha$ to a point $p'_\alpha$ whose distance to $p_\beta$ is
equal to the intrinsic distance along the surface, which exists by
\cite[Lemma~5]{holefilling}. Extending creases from $p'_\alpha$ to $o$, $q$, and
$p_\beta$ provides the crease pattern for this intermediate state. In fact we
parameterize the continuous family of crease patterns that folds this spanning
edge flat according to the location of $p'_\alpha$ along the segment between
$c_\alpha$ and $p_\alpha$, mapping the surface continuously to its flattened
state. 

\begin{lemma}
\label{lem:outout}
The \outout\ gadget has bounded size and stays between the top and bottom faces
during folding, while preserving intrinsic distances and avoiding crossings.
Further, the area supporting moving creases is also bounded and is proportional   
to the square of the gadget's height.
\end{lemma}

\begin{proof}
The \outout\ gadget has finite size; specifically, the introduced points
$p_\alpha$ and $p_\beta$ are within bounded projected distances
from point $o$ relative to $\theta$, $\alpha$, and $\beta$:
\begin{align}
(p_\alpha - o)\cdot u_\alpha =
\frac{1}{2}\csc\theta
\cot\left(\frac{\alpha+\beta-\theta}{4}\right)
\left(\cos(\alpha+\theta)-\cos\beta\right), \\
(p_\beta - o)\cdot u_\beta = 
\frac{1}{2}\csc\theta
\cot\left(\frac{\alpha+\beta-\theta}{4}\right)
\left(\cos(\beta+\theta)-\cos\alpha\right).
\end{align}

\noindent
where $u_\alpha$ is the unit direction along the bottom edge from $o$ adjacent
to $\alpha$ and similarly for $u_\beta$. These distances are bounded when
$0<\theta<\alpha+\beta\leq\pi$ as is required. Also, points $p'_\alpha$ and
$p_\beta$ remain between the top and bottom faces because they exist on the
spanning face creases; thus the entire gadget folds between the top and bottom
faces.

Isometry is satisfied by construction. It remains to show that
faces do not intersect. First, dihedral angles between adjacent faces in the
construction are always positive except in the flat state, so local crossing
does not occur between adjacent faces. Alternatively, the faces bounding the
$\alpha$ and $\beta$ spanning face creases cannot intersect each other as they
will always exist on opposite sides of a plane passing through $o$ and $q$, in
particular any such plane that also contains a base edge. Thus, the constructed
\outout\ gadget avoids crossings local to the gadget throughout the folding
motion.

The area supporting moving creases is shown in yellow in
Figure~\ref{fig:outout}.    
The area is bounded by product of the distance between $o$ and $q$ and
$(p_\alpha-o)\cdot u_\alpha + (p_\beta-o)\cdot u_\beta$, which is   
proportional to the square of the gadget's height.
\end{proof}

\begin{figure}
\includegraphics[width=\linewidth]{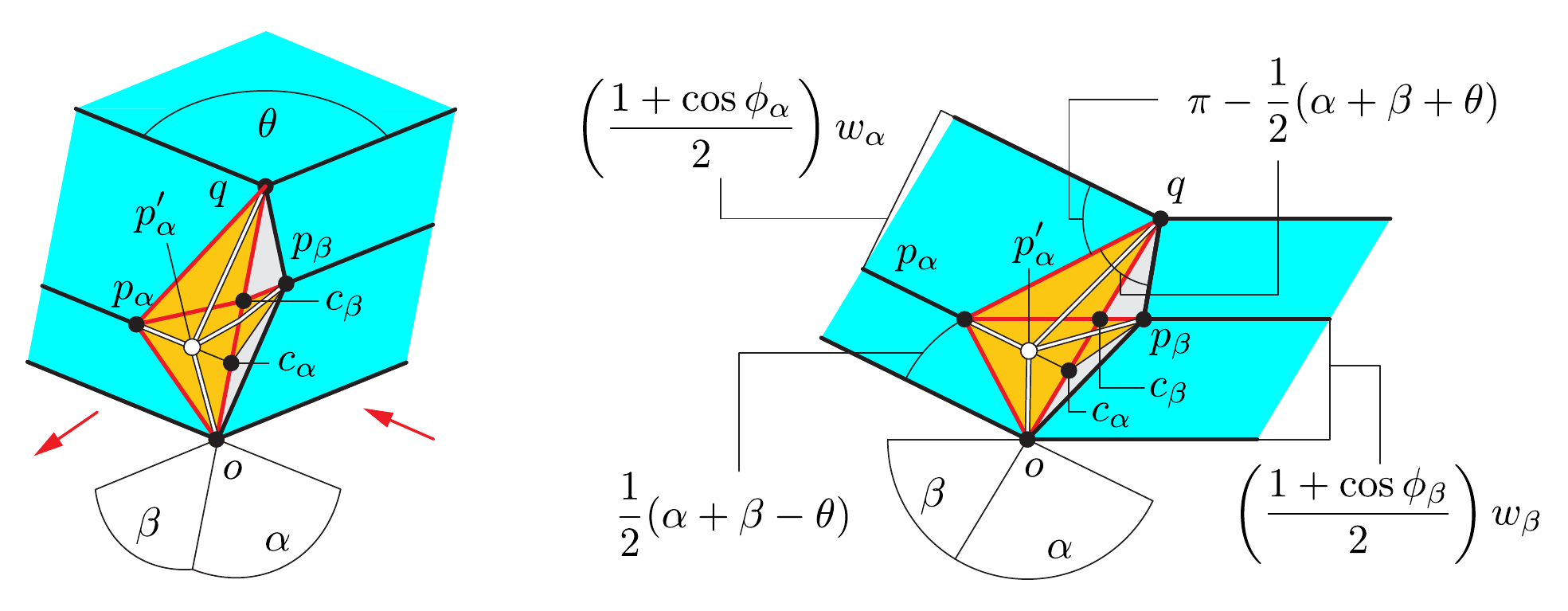}
\caption{Reference points and creases for the \inout\ gadget drawn on [Left]
the surface of a prismoid local to a spanning edge, and [Right] the development
of the spanning faces adjacent to the edge.} 
\label{fig:inout}
\end{figure}

\subsection{\inout\ Gadget}

Now we describe the construction of our \inout\ gadget, and show that it also
folds continuously while preserving intrinsic distances and avoiding crossings.
We again construct a flat folded state, introducing two crease pattern vertices,
moving one of which provides the desired folding motion. Figure~\ref{fig:inout}
corresponds to our construction described below. Again we construct the spanning
face creases for a prismoid spanning edge parameterized by $\theta$, $\alpha$,
and $\beta$. But this time we flatten the primary vertex $o$ locally using only one
crease so that one face collapses toward the convex side of the edge while the
other face collapses away. Without loss of generality, let the $\alpha$ side
collapse away from the convex side of the edge. We flatten the angle $\alpha +
\beta$ of material with one crease while keeping the bounding edges angle
$\theta$ apart, yielding a crease with angle $(\alpha+\beta-\theta)/2$ on the
$\alpha$ side and angle $(\alpha+\beta+\theta)/2$ on the $\beta$ side. We
terminate the crease when it meets the $\alpha$ spanning face crease at point
$p_\alpha$. Complete the crease pattern by adding a crease from $p_\alpha$ to
$q$. Again, trivial but tedious algebra confirms that this crease pattern is
always flat-foldable.

Similarly to the construction of the \outout\ gadget, we would like to identify
a quadrilateral of paper with contractive diagonals at intermediate folded
states. However, in this case there is no obvious single choice for where to
locate our stationary point $p_\beta$ along the $\beta$ spanning face crease.
Nevertheless, we continue using the same strategy as before. Again, we have a
single degree of freedom flat-foldable crease pattern whose base angle opens
monotonically from $\theta$ to $\alpha+\beta$ when unfolded. Fix this crease
pattern in some three-dimensional intermediate folded state so that the base
angle is strictly between $\theta$ and $\alpha+\beta$. But this time, cut the
folding along the segments from $p_\alpha$ to $o$ and $q$. Now when we rotate
the bent $\alpha$ and $\beta$ spanning faces toward each other. Every point on
the $\beta$ spanning face crease is separated by exactly the intrinsic distance
from point $p_\alpha$ before rotation. Additionally, point $c_\beta$ and any
point further from $o$ along the $\beta$ spanning face crease will be closer to
point $p_\alpha$ after rotation. Thus choosing any point $p_\beta$ along the ray
from $c_\beta$ would yield a quadrilateral with contractive diagonals, upon
which we could apply the construction in~\cite{holefilling} to find point
$p'_\alpha$ along the $\alpha$ spanning face crease that results in an isometry.
However, we cannot choose any such point. Consider for example point $c_\beta$.
When $\theta$ is less than $\pi/2$, $c_\beta$ can penetrate the bent $\alpha$
spanning face which we cannot allow. Thus we must choose some point along the
$\beta$ spanning face crease for which intersection does not occur.

With the \outout\ gadget, we were able to argue that the bent $\alpha$ and
$\beta$ spanning faces do not interact with each other by identifying a
separating plane. We use that same strategy to pick point $p_\beta$. Let
$p_\beta$ be the point on the $\beta$ spanning face crease such that the angle
at $q$ bounded by $p_\alpha$ and the top face edge on the $\alpha$ side equals
the angle at $q$ between $p_\alpha$ and $p_\beta$. This choice ensures that the
faces bounding the top face edges do not overlap in the folded state, so the
plane containing the top edge on the $\alpha$ side and point $o$ will always
separate the bent $\alpha$ and $\beta$ spanning faces. Note that when
$\pi - \beta < \theta$, this choice of $p_\beta$ will lie on the $\alpha$ side
of the line from $o$ to $q$. In such cases, $c_\beta$ being further away also
avoids intersection, so we use it for $p_\beta$ instead. Now, having fixed
$p_\beta$ for our spanning edge folded to some intermediate state, we have a
quadrilateral hole with vertices $o$, $p_\alpha$, $q$, and $p_\beta$ with
contractive diagonals. We again extend the spanning face crease incident to
$p_\alpha$ to a point $p'_\alpha$ whose distance to $p_\beta$ is equal to the
intrinsic distance along the surface, which exists by
\cite[Lemma~5]{holefilling}, and extending creases from $p'_\alpha$ to $o$, $q$,
and $p_\beta$ provides the crease pattern for this intermediate state. We
parameterize the continuous family of crease patterns in the same way as the
\outout\ gadget, by the location of $p'_\alpha$ along the segment between
$c_\alpha$ and $p_\alpha$, mapping the surface continuously to its flattened
state.

\begin{lemma}
\label{lem:inout}
The \inout\ gadget has bounded size and stays between the top and bottom faces
during folding, while preserving intrinsic distances and avoiding crossings.
Further, the area supporting moving creases is also bounded and is proportional   
to the square of the gadget's height.
\end{lemma}

\begin{proof}
The \inout\ gadget has finite size; specifically, the introduced points
$p_\alpha$ and $p_\beta$ are within constant projected distances
from point $o$ relative to $\theta$, $\alpha$, and $\beta$:
\begin{align}
(p_\alpha - o)\cdot u_\alpha &=
\frac{1}{2}\csc\theta
\cot\left(\frac{\alpha+\beta-\theta}{2}\right)
\left(\cos\beta-\cos(\alpha+\theta)\right), \\
\begin{split}
(p_\beta - o)\cdot u_\beta &= 
\frac{1}{2}\csc\theta\max\left[\,\cot\beta\left(
\cos(\beta+\theta)-\cos\alpha\right)\right., \\
&\hspace{2.52cm}\left. 
\cot\theta\left(\cos(\beta-\theta)-\cos\alpha\right)
+2\cos\beta\sin\theta
\right].
\end{split}
\end{align}

\noindent
where $u_\alpha$ is the unit direction along the bottom edge from $o$ adjacent
to and $\alpha$ and similarly for $u_\beta$. These distances are bounded when
$0<\theta<\alpha+\beta\leq\pi$ as is required. The remaining argument is
identical to the proof of Lemma~\ref{lem:outout}.

The area supporting moving creases is shown in yellow in Figure~\ref{fig:inout}.
The area is bounded by product of the distance between $o$ and $q$ and 
$(p_\alpha-o)\cdot u_\alpha + (p_\beta-o)\cdot u_\beta$, which is   
proportional to the square of the gadget's height.
\end{proof}


\section{Slicing}
\label{sec:slice}

In this section, we show how to slice our polyhedral manifold into
\emph{prismoids}
so techniques from the previous section can be applied. Once sliced, we can
collapse the subset in each slab separately. Because we are not restricting our
input to be homeomorphic to a sphere, components within a slab might not be
prismoids, but instead subsets of prismoids (i.e., missing faces). To deal with
this generalization, we define a \emph{prismoidal wall} to be a (non-strict)
subset of the spanning faces of some prismoid. Further, we define a
\emph{prismoidal slab} to be a finite set of disjoint prismoidal walls spanning
two planes, where each such plane is a \emph{base} of the slab.

We will slice prismoidal slabs along \emph{slice planes}, planes parallel to the
base strictly between the top and bottom of the slab. Subdividing a prismoidal
slab along a slice plane results in two prismoidal slabs with smaller height
than the original. Slicing will occur in multiple stages. First, we show how to
break a polyhedral manifold into countably infinitely many prismoidal slabs.
Next, we further subdivide each slab so the prismoidal walls in a slab have
disjoint projections onto the slab's base. Finally, we slice prismoidal walls
one last time in order to accommodate the local bounds for flattening spanning
edges and faces using the gadgets described in Section~\ref{sec:collapse}.

\subsection{Prismoidal Slab Decomposition}

\begin{lemma}
\label{thm:slice}
Any polyhedral manifold can be decomposed into a countably infinite set of
prismoidal slabs.
\end{lemma}
\begin{proof}

Orient the polyhedron so that each vertex has unique projection along some axis,
which will be the case for a generic choice of axis. Then slice a plane through
each vertex normal to that axis; see Figure~\ref{fig:manifoldslice} [Left].
This division decomposes the polyhedral manifold into a set of prismatoidal
slabs (not necessarily prismoidal slabs as faces adjacent to a sliced vertex may
be triangles). For every non-prismoidal slab, slice along the bisecting plane
between its top and bottom plane. Because each non-prismoidal slab contains
exactly one vertex by construction, bisecting them yields one prismoidal slab
and one non-prismoidal slab. Recursively bisecting all non-prismoidal slabs in
this way will decompose the polyhedral manifold into a countably infinite set of
prismoidal slabs, with slab heights approaching zero near each vertex.
Figure~\ref{fig:manifoldslice} [Right] shows this division applied to a
prismatoid with one vertex that is not degree-three. \end{proof}

\begin{figure}
\includegraphics[width=\linewidth]{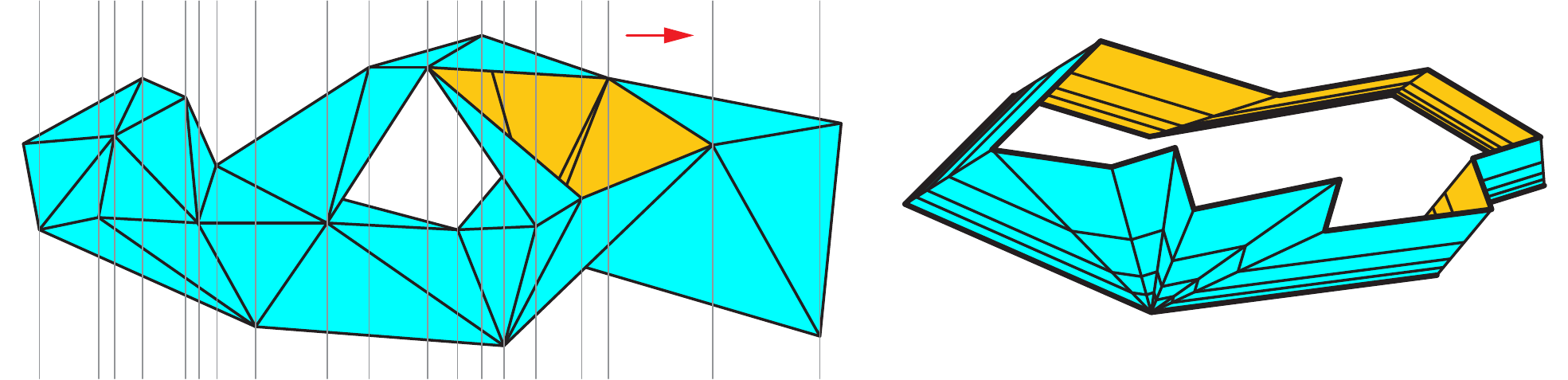}
\caption{[Left] Slicing a polyhedral manifold through vertices along planes
normal to a direction (indicated by the red arrow) onto which vertices have unique projection. [Right]
Slicing a prismatoid with one vertex that is not degree-three into an infinite
set of prismoids.}
\label{fig:manifoldslice}
\end{figure}

\subsection{Projection Disjoint Decomposition}

We call a prismoidal slab to be \emph{projection disjoint} if nonadjacent faces
in the slab do not overlap in the projection of the slab onto its base.

\begin{lemma}
\label{thm:separate}
Any prismoidal slab can be decomposed into a finite set of prismoidal slabs
that are each projection disjoint.
\end{lemma}

\begin{proof}

By definition, the prismoidal walls in a prismoidal slab are disjoint, and
because each face spans the top and bottom planes, each face has finite slope.
Let $\psi$ be the smallest tilt angle of any face in the prismoidal slab, and
let $s$ be the shortest distance between any vertex and a non-adjacent edge in
the same plane; see Figure~\ref{fig:projection}. Slice the prismoidal slab into
a set of shorter prismoidal slabs, each with height no greater than $(s/2)\sin
\psi$. Then the width of the projection of any face of the new slabs onto the
base will be no larger than $s/2$. Because $s$ is the minimum distance between a
vertex and an non-adjacent edge, these prismoidal slabs must be projection
disjoint. \end{proof}

\begin{figure}
\includegraphics[width=\linewidth]{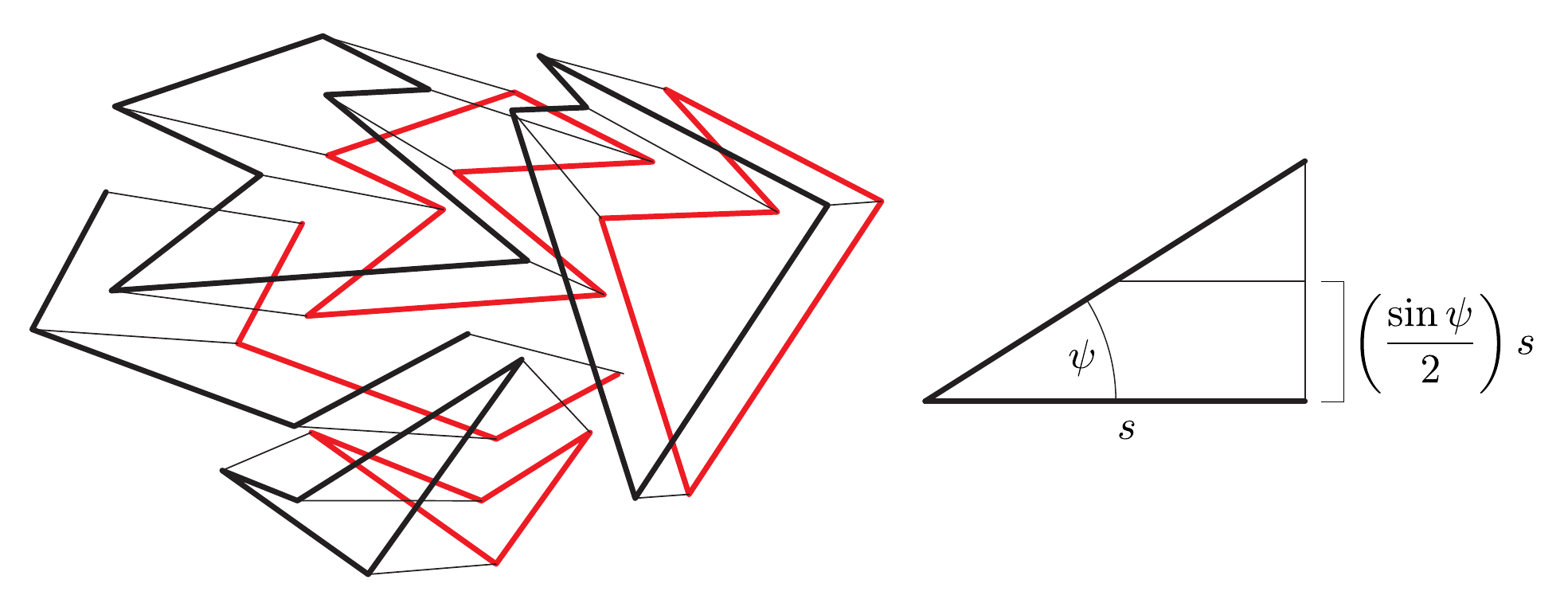}
\caption{[Left] A prismoidal slab that is not projection disjoint. [Right]
Calculating a split height that will allow decomposition into a set of uniform 
height slabs that are projection
disjoint. Angle $\psi$ is the smallest angle of a prismoidal spanning face
relative to the base and $s$ is the shortest distance between a vertex and a
nonadjacent edge in either the top or bottom planes.}
\label{fig:projection}
\end{figure}

\subsection{Flattening Projection Disjoint Prismoidal Slabs}

\begin{lemma}
\label{thm:local-bounds}

Any projection-disjoint prismoidal slab can be continuously flattened while
preserving intrinsic distances and avoiding crossings. Further, the area
supporting moving creases can be made arbitrarily small.

\end{lemma}

\begin{proof}
Our approach will be to use the \inout\ and \outout\ gadgets from
Section~\ref{sec:collapse} to collapse the prismoidal walls contained in the
prismoidal slab. However, there may not be room to construct the gadgets if the
spanning edges are too close together. To ensure spanning edges are well
separated, we slice the prismoidal slab one last time. The proofs of
Lemmas~\ref{lem:outout} and \ref{lem:inout} show that our spanning edge gadgets
are local to their spanning edge, existing within a distance proportional to the
height of the gadget and Section~\ref{sec:PrismoidOverview} specifies how these gadgets can be assigned to a prismoid. Further, Section~\ref{sec:facecollapse} bounds the
extension of spanning faces outside their projection onto the base, again within
a distance proportional to the height of the gadget. Slicing the slab in half
will reduce the reach of each edge gadget and face extension in half, while the
distances between them will stay fixed. Thus we can decompose the projection
disjoint prismoidal slab into a finite set of prismoidal slabs that each has
room to fold the gadgets at each spanning edge, and collapse each spanning face.
Constructing these gadgets, we can flatten them while preserving intrinsic
distances and avoiding crossing because we have sliced such that non-local
interactions do not occur.

The proofs of Lemmas~\ref{lem:outout} and \ref{lem:inout} also show that for
each gadget, the area supporting moving creases is proportional to the
height of the prismoidal wall squared. Thus, we can reduce this area arbitrarily
by further subdivision. 
\end{proof}


\section{Flattening Polyhedral Manifolds}
\label{sec:proof}

Knowing how to locally flatten prismoids using the gadgets from
Section~\ref{sec:collapse}, and how to split polyhedral manifolds into
projection-disjoint prismoidal slabs whose geometry is well-separated relative
to slab height from Section~\ref{sec:slice},
the proof of Theorem~\ref{thm:flatten} follows directly:

\begin{proof}[Proof of Theorem~\ref{thm:flatten}] 

Slice the prismoidal manifold into a countably infinite set of prismoidal slabs using the
construction in Lemma~\ref{thm:slice}. Then decompose each prismoidal slab into
projection disjoint prismoidal slabs with well-separated geometry according to
the constructions in Lemma~\ref{thm:separate} and Lemma~\ref{thm:local-bounds}.
Then by Lemma~\ref{thm:local-bounds}, we can continuously and independently flatten each slab while preserving intrinsic distances and avoiding crossings within each
slab. The flattening motion of each slab brings together the top and bottom
without transverse translation perpendicular to the axis, while
Lemmas~\ref{lem:outout} and \ref{lem:inout} guarantee that geometry within each
slab stays between the slab's top and bottom planes during folding; so crossing
cannot occur between slabs. Lastly, if we want to bound the area supporting
moving creases below any positive value, Lemma~\ref{thm:local-bounds} guarantees
that we can with further subdivision.
\end{proof}


\section{Conclusion}

In this paper, we showed how to continuously flatten finite polyhedral manifolds
using countably many creases with only finitely many accumulation lines. The
obvious open problem is whether continuous flattening is possible with only
finitely many creases at each time (still, of course, with movable creases that
slide over a 2D region of points over time). In the other direction,
perhaps our approach could be generalized to polyhedral manifolds with countably
many vertices, edges, and faces. Such a result may require a more general model
of what folding means for countably many creases (say, directly generalizing
\cite[Chapter~11]{GFALOP}), which is another interesting direction for pursuit;
our current model is very specific to our slice-based approach.

Our result is tight in a couple of senses.
We cannot hope to flatten surfaces with a positive area of nonflat points
(e.g., smooth surfaces like a sphere), even instantaneously, because a
flat folded state would need to be creased at all of those points.
(A more reasonable model in this case is contractive folding; see
\cite{SphereWrapping_CGTA}.)
Similarly, we cannot hope to further flatten from a plane to a line (or point)
without creases becoming everywhere-dense.

\bibliography{flattening}
\bibliographystyle{alpha}

\end{document}